\documentclass[conference,a4paper,10pt]{IEEEtran}
\IEEEoverridecommandlockouts
\usepackage{cite}
\ifCLASSINFOpdf
\else
\fi
\usepackage{amsthm,amsmath,amssymb,amsfonts}
\usepackage{mathtools}
\mathtoolsset{showonlyrefs}
\usepackage{algorithm,algorithmic}

\ifCLASSOPTIONcompsoc
 \usepackage[caption=false,font=normalsize,labelfont=sf,textfont=sf]{subfig}
\else
 \usepackage[caption=false,font=footnotesize]{subfig}
\fi
\usepackage{color}
\usepackage{nicefrac}

\usepackage{tikz}
\usetikzlibrary{arrows,automata,decorations.pathreplacing, mindmap,backgrounds}

\newtheorem{lemma}{Lemma}
\newtheorem{definition}{Definition}
\newtheorem{corollary}{Corollary}
\newtheorem{remark}{Remark}

\begin{document}


%
\title{On the Construction of Jointly Superregular Lower Triangular Toeplitz Matrices}



%
\author{\IEEEauthorblockN{Jonas Hansen\IEEEauthorrefmark{1}\IEEEauthorrefmark{2},
Jan {\O}stergaard\IEEEauthorrefmark{1},
Johnny Kudahl\IEEEauthorrefmark{2}, and
John H. Madsen\IEEEauthorrefmark{2}}
\IEEEauthorblockA{\IEEEauthorrefmark{1}Department of Electronic Systems, Aalborg University, Denmark\\}
\IEEEauthorblockA{\IEEEauthorrefmark{2}Bang \& Olufsen A/S, Struer, Denmark\\}
\IEEEauthorblockA{\{jh, jo\}@es.aau.dk, \{joy, jhm\}@bang-olufsen.dk}
\thanks{The work of J. Hansen was partially supported by The Danish National Innovation Foundation, Grant No. 4135-00131B.}
\thanks{The work of J. {\O}stergaard was partially supported by VILLUM FONDEN Young Investigator Programme, Project No. 10095.}
\thanks{\textcopyright\ 2016 IEEE. Personal use of this material is permitted. Permission from IEEE must be obtained for all other uses, in any current or future media, including reprinting/republishing this material for advertising or promotional purposes, creating new collective works, for resale or redistribution to servers or lists, or reuse of any copyrighted component of this work in other works.}
}


\maketitle

\begin{abstract}
Superregular matrices have the property that all of their submatrices, which can be full rank are so.
Lower triangular superregular matrices are useful for e.g., maximum distance separable convolutional codes as well as for~(sequential) network codes.
In this work, we provide an explicit design for all superregular lower triangular Toeplitz matrices in~$\mathrm{GF}(2^p)$ for the case of matrices with dimensions less than or equal to~$5\times 5$.
For higher dimensional matrices, we present a greedy algorithm that
finds a solution provided the field size is sufficiently high.
We also introduce the notions of jointly superregular and product
preserving jointly superregular matrices, and extend our explicit
constructions of superregular matrices to these cases.
Jointly superregular matrices are necessary to achieve optimal
decoding capabilities for the case of codes with a rate lower
than~$\nicefrac{1}{2}$, and the product preserving property is
necessary for optimal decoding capabilities in network recoding.
\end{abstract}


%
\IEEEpeerreviewmaketitle




\section{Introduction}
Wireless networks are used more and more for the streaming of audio and video data.
Generally, wireless packet--based streaming requires some amount of forward erasure correction in order to cope with packet erasures and latency constraints.
In a streaming context, erasure correcting codes and reliable transport protocols have been investigated in e.g.,~\cite{10.1109/ITCC.2001.918813,6620377,4427233,6284055}.
Erasure correcting codes are either applied as a block code on consecutive blocks of the incoming data or as a convolutional code that sequentially process the incoming data packets.
If the block code is lower triangular, it can be used sequentially on the incoming data in the same manner as a convolutional code.
Decoding can also be done sequentially as data packets are received, and, thus, the latency can be kept low.
If further coding is allowed within the network, and not only at the edges, it is usually referred to as network coding.
Besides enhanced reliability, network coding can offer increased throughput and security and has been successfully applied in various communication scenarios~\cite{6692495,5934877,6825089,7263352}.

Convolutional codes or lower triangular block codes may be constructed using a random linear code.
One of the benefits of random linear codes is simplicity, e.g., with respect to coordination between nodes.
Furthermore, for large field sizes and code dimensions, optimal decoding capabilities can often be proven at least asymptotically.
On the other hand, for small field sizes and small code dimensions, it is generally hard to guarantee optimal decoding capabilities, and the need for coordination usually implies that the resulting codes, if used as network codes, suffer from high overhead requirements~\cite{5963013}.

Coding matrices in small dimensions are of great interest for streaming applications.
The advantage of using small matrices are twofold.
First, they can be decoded with generic decoding algorithms such as Gaussian elimination, even on embedded devices, despite the cubic complexity of the algorithm.
Second, the small dimension allows for the construction of coding matrices that are guaranteed to be optimal in the non--asymptotic regime, and with memory requirements and field sizes that are feasible for encoding and decoding on embedded devices.
The implementation of~$\mathrm{GF}(2^p)$ arithmetic is also straightforward on digital devices, since they are based on binary processors.
This makes it feasible to implement high--performance~$\mathrm{GF}(2^p)$ arithmetic.
It is particular useful to use elements of~$\mathrm{GF}(2^8)$ since they can each be represented exactly by a single byte.

Both convolutional codes and lower triangular block codes may be constructed from lower triangular matrices.
In the latter case, we show in Fig.~\ref{fig:matrix_structure}, examples of rate~$\nicefrac{1}{2}$ and rate~$\nicefrac{1}{3}$ codes obtained by concatenating two or three lower triangular matrices, respectively.
In particular, let~$A$ be an~$n \times k$ coding matrix, where~$A$ can e.g., be illustrated as in Figs.~\ref{fig:matrix_structure}\subref{fig:matrix_structure_4x8_channel} and~\ref{fig:matrix_structure}\subref{fig:matrix_structure_4x12_channel}.
The rate of the code is given by~$k/n$.
Let~$S$ be the~$k \times l$ source data matrix, and let~$C$ be the~$n \times l$ coded data matrix, i.e., the output of the error correcting code.
Then~$C = AS$, which implies that the~$k$ source vectors of dimension~$l$ are encoded into~$n$ coded vectors each of length~$l$.
The matrices shown in Figs.~\ref{fig:matrix_structure}\subref{fig:matrix_structure_4x8_math} and~\ref{fig:matrix_structure}\subref{fig:matrix_structure_4x12_math}, contain the same rows as those in Figs.~\ref{fig:matrix_structure}\subref{fig:matrix_structure_4x8_channel} and~\ref{fig:matrix_structure}\subref{fig:matrix_structure_4x12_channel}.
However, the rows are ordered differently to better illustrate that the source vectors can be processed sequentially as they appear.
The use of the identity matrix as a code matrix yields a systematic code.
The benefit of using two concatenated~$m \times m$ coding matrices instead of one~$2m \times m$ matrix is twofold.
First, the entire coding matrix needs to preserve the low latency property, this is straightforward for the two square matrices by having them be lower triangular.
This property is not well defined for a tall matrix.
Second, in a multipath network the two square matrices may be used on different paths.
Splitting a tall matrix and using it in two different paths in a network is not desirable.

If an~$m \times m$ lower triangular matrix is superregular, then it is also an optimal block code, i.e., it has optimal decoding capabilities~\cite{1023698}.
A lower triangular matrix is superregular, if and only if all of its proper submatrices are non--singular~\cite{1580796}.
It was shown in~\cite{1580796} that MDS convolutional codes can be constructed from lower triangular superregular matrices.
Thus, it is of great interest to find a way to construct superregular lower triangular matrices in small dimensions and with small field sizes.
This is, however, an open problem.
In~\cite{1580796}, a few of such matrices were shown without providing insights to how they were obtained.
In~\cite{aissen1952}, an explicit construction for superregular~(totally positive) matrices was provided for real and complex fields.
This construction can easily be extended to very large prime fields, which is impractical.
In~\cite{Almeida20132145} a new class of lower block triangular matrices that are superregular over a sufficiently large field was presented.

In this paper, we provide an explicit design for all superregular lower triangular Toeplitz matrices in~$\mathrm{GF}(2^p)$ for the case of matrices with dimensions less than or equal to~$5\times 5$.
For general dimensions, we propose a greedy approach to design the lower triangular superregular Toeplitz matrices.

By concatenating the identity matrix and~$m-1$ code matrices, a rate~$1/m$ code is obtained.
Codes with a rate lower than~$\nicefrac{1}{2}$ are of concern in various applications such as audio/video streaming.
For example, it may be used in a streaming context when the underlying erasure channel suffers from a significant amount of erasures or in one--to--many scenarios such as broadcast erasure channels with limited feedback options.
Unfortunately, even if all the~$m-1$,~$m > 2$ individual code matrices are superregular, it is not guaranteed that their concatenation with the identity matrix yields an optimal~$1/m$ rate code.
To this end, we introduce the notion of jointly superregular matrices.
The use of two jointly superregular matrices maximizes the decoding capabilities, see Definition~\ref{def:jointly_super}.
With this stronger notion of superregularity, optimal decoding capabilities can be obtained for any~$1/m$ rate code.
We provide explicit constructions for such lower triangular matrices in small dimensions and any field~$\mathrm{GF}(2^p)$.

\begin{figure}[tb]
\newcommand{\CodeStructureLength}{7.7mm}
\subfloat[]{  \raisebox{-\height}{\begin{tikzpicture}[>=stealth', ,shorten >=1pt, ,shorten <=1pt, auto, semithick, , scale=0.5, every node/.style={scale=0.5}]
    \tikzstyle{every state}=[rectangle, fill=white,text=black,minimum height=\CodeStructureLength, minimum width=\CodeStructureLength, node distance=\CodeStructureLength]

    \node[state]  (L21) {1};
    \node[state, below of=L21, xshift=\CodeStructureLength] (L22) {1};
    \node[state, below of=L22, xshift=\CodeStructureLength] (L23) {1};
    \node[state, below of=L23, xshift=\CodeStructureLength] (L24) {1};

    \node[state, below of=L21, yshift=-3*\CodeStructureLength] (L25) {1};
    \node[state, below of=L25, draw=none, xshift=\CodeStructureLength] (L26a) {1};
    \node[state, left of=L26a, draw=none] (L26b) {$\omega^{i_1}$};
    \node[state, right of=L26b, fill=none, minimum width=2*\CodeStructureLength, xshift=-\CodeStructureLength/2] (L26) {};

    \node[state, below of=L26a, draw=none, xshift=\CodeStructureLength] (L27a) {1};
    \node[state, left of=L27a, draw=none] (L27b) {$\omega^{i_1}$};
    \node[state, left of=L27b, draw=none] (L27c) {$\omega^{i_2}$};
    \node[state, right of=L27c, fill=none, minimum width=3*\CodeStructureLength] (L27) {};

    \node[state, below of=L27a, draw=none, xshift=\CodeStructureLength] (L28a) {1};
    \node[state, left of=L28a, draw=none] (L28b) {$\omega^{i_1}$};
    \node[state, left of=L28b, draw=none] (L28c) {$\omega^{i_2}$};
    \node[state, left of=L28c, draw=none] (L28d) {$\omega^{i_3}$};
    \node[state, right of=L28c, fill=none, minimum width=4*\CodeStructureLength, xshift=-\CodeStructureLength/2] (L28) {};

  \end{tikzpicture}}\label{fig:matrix_structure_4x8_math}}\hfill
\subfloat[]{  \raisebox{-\height}{\begin{tikzpicture}[>=stealth', ,shorten >=1pt, ,shorten <=1pt, auto, semithick, , scale=0.5, every node/.style={scale=0.5}]
    \tikzstyle{every state}=[rectangle, fill=white,text=black,minimum height=\CodeStructureLength, minimum width=\CodeStructureLength, node distance=\CodeStructureLength]

    \node[state]                (L11) {1};

    \node[state, below of=L11]  (L12) {1};

    \node[state, below of=L12, xshift=\CodeStructureLength] (L13) {1};

    \node[state, below of=L13, draw=none] (L14a) {1};
    \node[state, left of=L14a, draw=none] (L14b) {$\omega^{i_1}$};
    \node[state, right of=L14b, fill=none, minimum width=2*\CodeStructureLength, xshift=-\CodeStructureLength/2] (L14) {};

    \node[state, below of=L14a, xshift=\CodeStructureLength] (L15) {1};

    \node[state, below of=L15, draw=none] (L16a) {1};
    \node[state, left of=L16a, draw=none] (L16b) {$\omega^{i_1}$};
    \node[state, left of=L16b, draw=none] (L16c) {$\omega^{i_2}$};
    \node[state, right of=L16c, fill=none, minimum width=3*\CodeStructureLength] (L16) {};

    \node[state, below of=L16a, xshift=\CodeStructureLength] (L17) {1};

    \node[state, below of=L17, draw=none] (L18a) {1};
    \node[state, left of=L18a, draw=none] (L18b) {$\omega^{i_1}$};
    \node[state, left of=L18b, draw=none] (L18c) {$\omega^{i_2}$};
    \node[state, left of=L18c, draw=none] (L18d) {$\omega^{i_3}$};
    \node[state, right of=L18c, fill=none, minimum width=4*\CodeStructureLength, xshift=-\CodeStructureLength/2] (L18) {};

  \end{tikzpicture}}\label{fig:matrix_structure_4x8_channel}}\hfill
\subfloat[]{  \raisebox{-\height}{\begin{tikzpicture}[>=stealth', ,shorten >=1pt, ,shorten <=1pt, auto, semithick, , scale=0.5, every node/.style={scale=0.5}]
    \tikzstyle{every state}=[rectangle, fill=white,text=black,minimum height=\CodeStructureLength, minimum width=\CodeStructureLength, node distance=\CodeStructureLength]

    \node[state]  (L21) {1};
    \node[state, below of=L21, xshift=\CodeStructureLength] (L22) {1};
    \node[state, below of=L22, xshift=\CodeStructureLength] (L23) {1};
    \node[state, below of=L23, xshift=\CodeStructureLength] (L24) {1};

    \node[state, below of=L21, yshift=-3*\CodeStructureLength] (L25) {1};
    \node[state, below of=L25, draw=none, xshift=\CodeStructureLength] (L26a) {1};
    \node[state, left of=L26a, draw=none] (L26b) {$\omega^{i_{a_1}}$};
    \node[state, right of=L26b, fill=none, minimum width=2*\CodeStructureLength, xshift=-\CodeStructureLength/2] (L26) {};

    \node[state, below of=L26a, draw=none, xshift=\CodeStructureLength] (L27a) {1};
    \node[state, left of=L27a, draw=none] (L27b) {$\omega^{i_{a_1}}$};
    \node[state, left of=L27b, draw=none] (L27c) {$\omega^{i_{a_2}}$};
    \node[state, right of=L27c, fill=none, minimum width=3*\CodeStructureLength] (L27) {};

    \node[state, below of=L27a, draw=none, xshift=\CodeStructureLength] (L28a) {1};
    \node[state, left of=L28a, draw=none] (L28b) {$\omega^{i_{a_1}}$};
    \node[state, left of=L28b, draw=none] (L28c) {$\omega^{i_{a_2}}$};
    \node[state, left of=L28c, draw=none] (L28d) {$\omega^{i_{a_3}}$};
    \node[state, right of=L28c, fill=none, minimum width=4*\CodeStructureLength, xshift=-\CodeStructureLength/2] (L28) {};

    \node[state, below of=L21, yshift=-7*\CodeStructureLength] (L35) {1};
    \node[state, below of=L35, draw=none, xshift=\CodeStructureLength] (L36a) {1};
    \node[state, left of=L36a, draw=none] (L36b) {$\omega^{i_{b_1}}$};
    \node[state, right of=L36b, fill=none, minimum width=2*\CodeStructureLength, xshift=-\CodeStructureLength/2] (L36) {};

    \node[state, below of=L36a, draw=none, xshift=\CodeStructureLength] (L37a) {1};
    \node[state, left of=L37a, draw=none] (L37b) {$\omega^{i_{b_1}}$};
    \node[state, left of=L37b, draw=none] (L37c) {$\omega^{i_2}$};
    \node[state, right of=L37c, fill=none, minimum width=3*\CodeStructureLength] (L37) {};

    \node[state, below of=L37a, draw=none, xshift=\CodeStructureLength] (L38a) {1};
    \node[state, left of=L38a, draw=none] (L38b) {$\omega^{i_{b_1}}$};
    \node[state, left of=L38b, draw=none] (L38c) {$\omega^{i_{b_2}}$};
    \node[state, left of=L38c, draw=none] (L38d) {$\omega^{i_{b_3}}$};
    \node[state, right of=L38c, fill=none, minimum width=4*\CodeStructureLength, xshift=-\CodeStructureLength/2] (L38) {};

  \end{tikzpicture}}\label{fig:matrix_structure_4x12_math}}\hfill
\subfloat[]{  \raisebox{-\height}{\begin{tikzpicture}[>=stealth', ,shorten >=1pt, ,shorten <=1pt, auto, semithick, , scale=0.5, every node/.style={scale=0.5}]
    \tikzstyle{every state}=[rectangle, fill=white,text=black,minimum height=\CodeStructureLength, minimum width=\CodeStructureLength, node distance=\CodeStructureLength]

    \node[state]  (L21) {1};

    \node[state, below of=L21] (L22) {1};

    \node[state, below of=L22] (L23) {1};

    \node[state, below of=L23, xshift=\CodeStructureLength] (L24) {1};

    \node[state, below of=L24, draw=none] (L25a) {1};
    \node[state, left of=L25a, draw=none] (L25b) {$\omega^{i_{a_1}}$};
    \node[state, right of=L25b, fill=none, minimum width=2*\CodeStructureLength, xshift=-\CodeStructureLength/2] (L25) {};

    \node[state, below of=L25a, draw=none] (L26a) {1};
    \node[state, left of=L26a, draw=none] (L26b) {$\omega^{i_{b_1}}$};
    \node[state, right of=L26b, fill=none, minimum width=2*\CodeStructureLength, xshift=-\CodeStructureLength/2] (L26) {};

    \node[state, below of=L26a, xshift=\CodeStructureLength] (L27) {1};

    \node[state, below of=L27, draw=none] (L28a) {1};
    \node[state, left of=L28a, draw=none] (L28b) {$\omega^{i_{a_1}}$};
    \node[state, left of=L28b, draw=none] (L28c) {$\omega^{i_{a_2}}$};
    \node[state, right of=L28c, fill=none, minimum width=3*\CodeStructureLength] (L28) {};

    \node[state, below of=L28a, draw=none] (L29a) {1};
    \node[state, left of=L29a, draw=none] (L29b) {$\omega^{i_{b_1}}$};
    \node[state, left of=L29b, draw=none] (L29c) {$\omega^{i_{b_2}}$};
    \node[state, right of=L29c, fill=none, minimum width=3*\CodeStructureLength] (L29) {};

    \node[state, below of=L29a, xshift=\CodeStructureLength] (L30) {1};

    \node[state, below of=L30, draw=none] (L31a) {1};
    \node[state, left of=L31a, draw=none] (L31b) {$\omega^{i_{a_1}}$};
    \node[state, left of=L31b, draw=none] (L31c) {$\omega^{i_{a_2}}$};
    \node[state, left of=L31c, draw=none] (L31d) {$\omega^{i_{a_3}}$};
    \node[state, right of=L31c, fill=none, minimum width=4*\CodeStructureLength, xshift=-\CodeStructureLength/2] (L31) {};

    \node[state, below of=L31a, draw=none] (L32a) {1};
    \node[state, left of=L32a, draw=none] (L32b) {$\omega^{i_{b_1}}$};
    \node[state, left of=L32b, draw=none] (L32c) {$\omega^{i_{b_2}}$};
    \node[state, left of=L32c, draw=none] (L32d) {$\omega^{i_{b_3}}$};
    \node[state, right of=L32c, fill=none, minimum width=4*\CodeStructureLength, xshift=-\CodeStructureLength/2] (L32) {};

  \end{tikzpicture}}\label{fig:matrix_structure_4x12_channel}}
\caption{The matrix structure used throughout this paper.~(a) and~(b) are rate~$\nicefrac{1}{2}$ codes.~(c) and~(d) are rate~$\nicefrac{1}{3}$ codes.~(a) and~(c) are the matrices used in the lemmas. Whereas,~(b) and~(d) are the matrices used on a erasure channel in an implementation.~(c) should be constructed using the identity matrix and two jointly superregular lower triangular Toeplitz matrices.}
\label{fig:matrix_structure}
\vspace{-4mm} 
\end{figure}
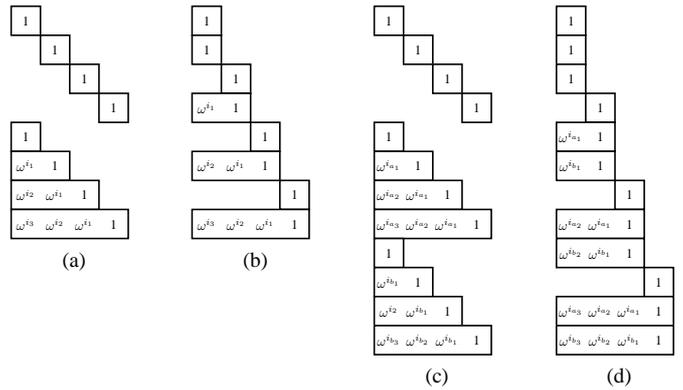

In ad--hoc and peer--to--peer networks, such as machine--to--machine communication or Internet of things, it is becoming more and more relevant to recode at intermediate nodes.
Recoding in network coding basically corresponds to multiplication of different coding matrices.
However, the resulting coding matrix obtained by multiplying two (jointly) superregular matrices is not guaranteed to be superregular.
We therefore introduce the notion of product preserving jointly superregular matrices.
In particular, given a pair of jointly superregular matrices, say~$A$ and~$B$, where~$A$ is used for encoding at the source and~$B$ is used at an intermediate node in the network to perform recoding.
Maximum decoding capabilities at the end--node is achieved if and only if~$AB$ or~$BA$ is superregular, which is guaranteed if~$A$ and~$B$ are product preserving jointly superregular matrices.
We provide a few explicit constructions for product preserving jointly superregular matrices in small dimensions and any field~$\mathrm{GF}(2^p)$.

\section{Superregular Matrices}
In a slightly different context the authors of~\cite{45291} define a dense matrix to be superregular if and only if every square submatrix is non--singular.
This definition of superregularity is extended in~\cite{1580796} to lower triangular matrices.
That is, a lower triangular matrix is superregular if and only if all of its proper submatrices are non--singular~\mbox{\cite[Definition 3.3]{1580796}}.
Let~$A$ be an~$m \times m$ lower triangular Toeplitz matrix with all the elements in the first column being non--zero.
Let~\mbox{$A' = A^{j_1,\dotsc,j_r}_{h_1,\dotsc,h_r}$} be a~$r \times r$ submatrix of~$A$.
Where~$A'$ is constructed using the rows and columns of~$A$ with indices~$j_1,\dotsc,j_r$ and~$h_1,\dotsc,h_r$, respectively~\mbox{\cite[Definition 3.2]{1580796}}.
Then,~$A'$ is a proper submatrix of~$A$ if and only if~\mbox{$1 \leq j_1 < j_2 < \dotsc < j_r \leq m$},~\mbox{$1 \leq h_1 < h_2 < \dotsc < h_r \leq m$} and~\mbox{$j_t \geq h_t, \forall t$}.
We adopt this notion of superregularity since it maximises the decoding capability~\cite{1023698}, when a superregular matrix is used in a code with rate~$\nicefrac{1}{2}$.
This notion of superregularity is somewhat different from the notion used in~\cite{7282860}.
Naturally, a code with rate greater than~$\nicefrac{1}{m}$ for some~$m$ can be generated through puncturing.

A code with rate~$\nicefrac{1}{3}$ can be constructed by using two jointly superregular matrices.
Naturally, two jointly superregular matrices are individually superregular.
The following definition describes the notion of joint superregularity.
The essential part of the definition is that any square submatrix formed by any combination of the two matrices that can be non--singular must also be non--singular.

\begin{definition}[Joint superregularity]
\label{def:jointly_super}
Two superregular~$t \times t$ matrices are said to be jointly superregular if and only if all of the proper submatrices of any~$t \times t$ matrix, formed by taking~$l = \{1,\dotsc,t-1\}$ and~$t-l$ rows from the two matrices, respectively,  are non--singular.
In the context of jointly superregular matrices, a proper submatrix is any square matrix that is not trivially rank deficient.
An~$m \times m$ matrix, when sorted by increasing row support\footnote{The support of a vector is equal to its number of non--zero elements.}, is said to be trivially rank deficient if the support of row~$i$,~$1 \leq i \leq m$, is less than~$i$.
A proper submatrix need not be triangular.
\hfill$\triangle$
\end{definition}

\begin{table}[tb]
\caption{The amount of~$5 \times 5$ superregular lower triangular Toeplitz matrices for different values of~$p$.}
\label{tab:sup_5}
\centering
\begin{tabular}{l|l|l|l}
$p$ & $f_p(\omega)$ & Lemma~\ref{lem:super}(iii) & Corollary~\ref{cor:sup_5} \\\hline
2 & $\omega^2 + \omega + 1$ & 0 & 0 \\\hline
3 & $\omega^3 + \omega + 1$ & 84 & 0 \\\hline
4 & $\omega^4 + \omega + 1$ & 17280 & 9 \\\hline
5 & $\omega^5 + \omega^2 + 1$ & 582180 & 2011 \\\hline
6 & $\omega^6 + \omega + 1$ & 12700800 & 76506 \\\hline
7 & $\omega^7 + \omega^3 + 1$ & 233847322 & 1234973 \\\hline
8 & $\omega^8 + \omega^4 + \omega^3 + \omega^2 + 1$ & 2000121984 & 17274832
\end{tabular}
\end{table}

\begin{definition}[Product preserving jointly superregular]
Two jointly superregular matrices are product preserving if and only if their product is a superregular matrix.
\hfill$\triangle$
\end{definition}

Let~$\Omega_{f_p}$ denote the set of roots of a primitive polynomial~$f_p$, which generates~$\mathrm{GF}(2^p)$. Let~\mbox{$\mathcal{I}_{n} = \{i_1,\dotsc, i_{n} | i_j \in \mathrm{GF}(2^p), i_j \neq 2^p-1, \forall j\}$}.
Let~\mbox{$\omega \in \Omega_{f_p}$} and let~\mbox{$\mathcal{A}_n^{\omega}$} denote the set of all~$n\times n$ superregular lower triangular Toeplitz matrices with their first column given by~\mbox{$[1,\omega^{i_1}, \omega^{i_2}, \dotsc, \omega^{i_{n-1}}]^T$}, where~\mbox{$(i_1,\dotsc, i_{n-1})\in \mathcal{I}_{n-1}$}.
Let~\mbox{$A_{n-1}\in \mathcal{A}_{n-1}^{\omega}$} and let~\mbox{$\phi_{\omega,i}\left({A}_{n-1}\right)$} denote an~$n\times n$ matrix obtained by extending~$A_{n-1}$ below and to the right by the row vector~\mbox{$[\omega^{i_{n-1}}, \omega^{i_{n-2}}, \dotsc, \omega^{i_{1}}, 1]$} and column vector~\mbox{$[0, \dotsc, 0, 1]^T$}, respectively, so that~\mbox{$\phi_{\omega,i}\left({A}_{n-1}\right)$} is lower triangular and Toeplitz.
Finally, let~$\psi_\omega\left(i_1,\dotsc,i_{n-1}\right)$ be an~$n \times n$ lower triangular Toeplitz matrix having the first column given by~$\left[1, \omega^{i_1},\dotsc,\omega^{i_{n-1}}\right]^T$.

Let~$\mathcal{B}_n^\omega$ denote the set of all pairs of jointly superregular matrices according to Definition~\ref{def:jointly_super}.
For two jointly superregular matrices,~$A_n$ and~$B_n$, we use the subscripts~$a$ and~$b$ to distinguish between their elements.
For~$(A_n, B_n) \in \mathcal{B}_n^{\omega}$, the first columns of~$A_n$ and~$B_n$ are given by~$[1,\omega^{i_{a_1}}, \dotsc, \omega^{i_{a_{n-1}}}]^T$ and~$[1,\omega^{i_{b_1}}, \dotsc, \omega^{i_{b_{n-1}}}]^T$, respectively, where~$(i_{a_1}, \dotsc, i_{a_{n-1}}, i_{b_1}, \dotsc, i_{b_{n-1}}) \in \mathcal{I}_{2n-2}$.

Let~$(A_n,B_n)=\phi_{\omega,i_a,i_b}(A_{n-1}, B_{n-1}) = (\phi_{\omega,i_a}(A_{n-1}), \phi_{\omega,i_b}(B_{n-1}))$ be the pair of~$n \times n$ matrices obtained by extending~$A_{n-1}$ and~$B_{n-1}$ using the straightforward generalization of the~$\phi$--operator for a single matrix.

In~\cite{7282860}, a construction of matrices that preserve superregularity after multiplication with block diagonal matrices was constructed.
In our case, the product of two superregular matrices is not guaranteed to be a superregular matrix.
Note that the multiplication (from the right) in~\cite{7282860} is different as the matrices have entries in different fields.
\begin{lemma}
Given~$A_n \in \mathcal{A}_n^\omega$, then~$\exists A_n' \in \mathcal{A}_n^\omega$ such that their product~$A_nA_n' \notin \mathcal{A}_n^\omega$.
\hfill$\triangle$
\end{lemma}
\begin{proof}
The proof follows easily from~\cite[Corollary 3.6]{1580796}. For any~$A_n \in \mathcal{A}_n^\omega$ then~$A_n^{-1} \in \mathcal{A}_n^\omega$ and it follows that~\mbox{$A_nA_n^{-1} = I_n \notin \mathcal{A}_n^\omega$}.
\end{proof}

Let~$\mathcal{C}_n^{\omega}$ denote the set of all pairs of~$n \times n$ product preserving jointly superregular lower triangular
Toeplitz matrices:
\begin{equation}
\mathcal{C}_n^{\omega} \triangleq \{ (A_n, B_n) \in \mathcal{B}_n^{\omega} : A_nB_n = B_nA_n \in \mathcal{A}_n^{\omega} \}, \omega \in \Omega_{f_p}.
\end{equation}

\section{Explicit construction of superregular and jointly superregular matrices}
In this section we first show methods for explicit construction of lower triangular Toeplitz superregular matrices of size~$n \times n$, where~$n \leq 5$.
Any matrix of size~$2 \times 2$ with~$i_1 \in \mathcal{I}_1$ is superregular over some~$\mathrm{GF}(2^p)$.
This follows easily from the definition since~$\omega^{i_1} \neq 0$.
In the following field operations on the elements of $\mathcal{I}_n$ are taken modulo~$2^p-1$.

\begin{lemma}
\label{lem:super}
Let~$\omega\in \Omega_{f_p}$ and $A_n = \psi_\omega(i_1,\dotsc,i_{n-1})$.
\begin{enumerate}
\item[i)]{Then~$A_3 \in \mathcal{A}_3^\omega$ if
and only if~$(i_1,i_2) \in \mathcal{I}_2$ and~$2i_1 \neq i_2$.}
\item[ii)]{Let~$A_{3} \in \mathcal{A}_3^\omega$ and~\mbox{$A_4 = \phi_{\omega,i_3}\left({A}_3\right)$}. Then~$A_4 \in \mathcal{A}_4^{\omega}$ if and only if,~\mbox{$(i_1,\dotsc, i_3) \in \mathcal{I}_{3}$} and satisfy:
\begin{align}
3i_1 \neq i_3, \quad i_1+i_2 \neq i_3, \quad 2i_2 \neq i_1+i_3.\label{eq:A4_eq}
\end{align}}
\item[iii)]{Let~$A_{4} \in \mathcal{A}_4^{\omega}$ and~\mbox{$A_5= \phi_{\omega,i_4}\left({A}_4\right)$}. Then~$A_5 \in \mathcal{A}_5^{\omega}$ if and only if,~\mbox{$(i_1,\dotsc, i_4) \in \mathcal{I}_4$} and satisfy:
\begin{equation}
\begin{split}
i_4 & \neq 2i_1 + i_2,\quad i_4 \neq i_1 + i_3, \\
i_4 & \neq 2i_2,\quad 2i_3 \neq i_2 + i_4,\quad i_2+i_3 \neq i_1 + i_4.\label{eq:A5_eq}
\end{split}
\end{equation}
and~$\omega$ and~$(i_1,\dotsc,i_4)$ jointly satisfy:
\begin{align}
\begin{split}
0 \neq & \omega^{2i_2+i_1} \oplus \omega^{i_2+i_3} \oplus \omega^{2i_1+i_3} \oplus \omega^{i_1+i_4}, \\
0 \neq & \omega^{2i_1+i_4} \oplus \omega^{i_2+i_4} \oplus \omega^{3i_2} \oplus \omega^{2i_3}, \\
0 \neq & \omega^{2i_1+i_2} \oplus \omega^{i_1+i_3} \oplus \omega^{2i_2} \oplus \omega^{i_4}, \\
0 \neq & \omega^{2i_1+i_2} \oplus \omega^{4i_1} \oplus \omega^{2i_2} \oplus \omega^{i_4}.\label{eq:A5_eq_2}
\end{split}
\end{align}}
\end{enumerate}
\hfill$\triangle$
\end{lemma}

Lemma~\ref{lem:super} (whose proof can be found in the Appendix) provides necessary and sufficient conditions for superregularity.
For the case of only sufficient conditions, the four non--trivial equations in~\eqref{eq:A5_eq_2} can be replaced by a single equation as shown in Corollary~\ref{cor:sup_5}.
Table~\ref{tab:sup_5} shows the number of superregular lower triangular Toeplitz matrices.

\begin{corollary}
\label{cor:sup_5}
Let~$\omega\in \Omega_{f_p}$,~$A_{4} \in \mathcal{A}_4^\omega$, and~$A_5 = \phi_{\omega,i_4}\left({A}_4\right)$.
Then~$A_5 \in \mathcal{A}_5^{\omega}$ if~$(i_1,\dotsc,
i_4) \in \mathcal{I}_4$ and satisfy \eqref{eq:A5_eq} and~$\omega$ and~$(i_1,\dotsc,i_4)$ jointly satisfy:
\begin{align}
0 \neq & \omega^{a \cdot i_1 + b \cdot i_2} \oplus \omega^{c \cdot i_1 + i_4} \oplus \omega^{d \cdot i_1 + e \cdot i_3} \oplus \omega^{f \cdot i_2 + g \cdot i_3 + h \cdot i_4},
\end{align}
where:
\mbox{$a,c,e \in \{0,1,2\}$},
\mbox{$b \in \{1,2,3\}$},
\mbox{$d \in \{0,1,2,4\}$},
\mbox{$f \in \{1,2\}$}, and
\mbox{$g,h \in \{0,1\}$}.
\hfill$\triangle$
\end{corollary}

\begin{remark}
Let~$\omega\in \Omega_{f_p}$.
If~$\psi_\omega\left(i_1,\dotsc,i_{n-1}\right) \in \mathcal{A}_n^{\omega}$ then~$\psi_{\omega'}\left(i_1,\dotsc,i_{n-1}\right) \in \mathcal{A}_n^{\omega}, \forall \omega'\in \Omega_{f_p}$.
\end{remark}

The two lemmas below,~\ref{lem:join_2} and~\ref{lem:join_3}, list the necessary and sufficient conditions for constructing jointly superregular lower triangular Toeplitz matrices of size~$2 \times 2$ and~$3 \times 3$, respectively.
Furthermore, Lemma~\ref{lem:join_2} also define a necessary condition for constructing jointly superregular lower triangular Toeplitz matrices of size~$n \times n$.

\begin{lemma}
\label{lem:join_2}
Let~$\omega\in \Omega_{f_p}$. For~$n = 2$,~$(A_2, B_2)\in \mathcal{B}_2^\omega$ if and only if,~$(i_{a_1},i_{b_1})\in \mathcal{I}_2$ and~$i_{a_1} \neq i_{b_1}$. For any~$n>1$, $(A_n, B_n) \notin \mathcal{B}_n^\omega$, if~$\exists j \in \{1,\dotsc,n-1\}$ such that~$i_{a_j} = i_{b_j}$.
\hfill$\triangle$
\end{lemma}
\begin{proof}
The determinant of the~$2\times 2$ submatrix~$\left[\frac{{A_n}^j_{1,j}}{ {B_n}^j_{1,j}} \right]$ is given by~$\omega^{i_{a_{j-1}}} \oplus \omega^{i_{b_{j-1}}}, \forall 1 < j \leq n$, and is only zero when~\mbox{$i_{a_{j-1}} = i_{b_{j-1}}$}.
\end{proof}

\begin{lemma}
\label{lem:join_3}
Let~$\omega\in \Omega_{f_p}$, and let~$(A_2, B_2) \in \mathcal{B}_2^\omega$. Let~$(A_3, B_3) = \phi_{\omega,i_{a_2},i_{b_2}}\left(A_2, B_2\right)$.
Then~$(A_3, B_3) \in \mathcal{B}_3^\omega$ if and only if,~$(i_{a_1},i_{a_2},i_{b_1},i_{b_2}) \in \mathcal{I}_{4}$ and satisfy:
\begin{equation}
i_{a_1} + i_{b_1} \neq i_{a_2},~~~~
i_{a_1} + i_{b_1} \neq i_{b_2},~~~~
i_{a_1} + i_{b_2} \neq i_{a_2} + i_{b_1}
\end{equation}
and~$\omega$ and~$(i_{a_1},i_{a_2},i_{b_1},i_{b_2})$ jointly satisfy:
\begin{align}
&&0 \neq & \omega^{i_{a_2}} \oplus \omega^{i_{b_2}} \oplus \omega^{i_{a_1} + i_{b_1}} \oplus \omega^{2i_{a_1}},\\
~~~~~~~~~~~&&0 \neq & \omega^{i_{a_2}} \oplus \omega^{i_{b_2}} \oplus \omega^{i_{a_1} + i_{b_1}} \oplus \omega^{2i_{b_1}}.&~~~~~~~~~~~\triangle 
\end{align}
\end{lemma}
The proof of Lemma~\ref{lem:join_3} uses a similar technique as used in the proof of Lemma~\ref{lem:super}, and it has therefore been omitted.

\begin{remark}
Let~$A_n \in \mathcal{A}_n^\omega$, where~$n > 1$, then~\mbox{$(A_n,A_n^{-1}) \notin \mathcal{B}_n^\omega$}.
\end{remark}
\begin{proof}
The proof follows from the fact that~$i_1$ for~$A_n$ is equal to~$i_1$ for~$A_n^{-1}$, which does not satisfy Lemma~\ref{lem:join_2}.
\end{proof}

Jointly superregular matrices of size~$2 \times 2$ are always product preserving.
The following lemma provides necessary and sufficient conditions for product preserving jointly superregular lower triangular Toeplitz matrices of size~$3 \times 3$ and~$4 \times 4$.

\begin{lemma}
\label{lem:product_prev}
Let~$\omega\in \Omega_{f_p}$.
\begin{enumerate}
\item[i)]{Let~$(A_3, B_3) \in \mathcal{B}_3^\omega$.
Then~$(A_3, B_3) \in \mathcal{C}_3^\omega$ if and only if,~$\omega$ and~$(i_{a_1},i_{a_2},i_{b_1},i_{b_2})$ jointly satisfy:
\begin{align}
0 \neq & \omega^{i_{a_2}} \oplus \omega^{i_{b_2}} \oplus \omega^{i_{a_1} + i_{b_1}},\\
0 \neq & \omega^{i_{a_2}} \oplus \omega^{i_{b_2}} \oplus \omega^{i_{a_1} + i_{b_1}} \oplus \omega^{2i_{a_1}} \oplus \omega^{2i_{b_1}}.
\end{align}}
\item[ii)]{Let~$(A_4, B_4) \in \mathcal{B}_4^\omega$.
Then~$(A_4, B_4) \in \mathcal{C}_4^\omega$ if and only if,~$\omega$ and~$(i_{a_1},\dotsc,i_{a_3},i_{b_1},\dotsc,i_{b_3})$ jointly satisfy:
\begin{align}
0 \neq & \omega^{i_{b_1}+i_{a_3}} \oplus \omega^{i_{b_3}+i_{a_1}} \oplus \omega^{i_{a_1}+i_{a_3}} \oplus \omega^{i_{b_2}+2i_{a_1}} \\
& \oplus \omega^{2i_{b_1}+i_{a_2}} \oplus \omega^{2i_{b_2}} \oplus \omega^{2i_{b_1}+2i_{a_1}} \oplus \omega^{2i_{a_2}} \\
& \oplus \omega^{i_{b_1}+i_{b_3}} \oplus \omega^{i_{b_1}+i_{b_2}+i_{a_1}} \oplus \omega^{i_{b_1}+i_{a_1}+i_{a_2}}, \label{eq:prod_4_1}\\
0 \neq & \omega^{i_{a_3}} \oplus \omega^{i_{b_3}} \oplus \omega^{i_{b_1}+i_{a_2}} \oplus \omega^{i_{b_2}+i_{a_1}} \oplus \omega^{i_{b_1}+2i_{a_1}} \\
& \oplus \omega^{2i_{b_1}+i_{a_1}} \oplus \omega^{3i_{b_1}} \oplus \omega^{3i_{a_1}}, \label{eq:prod_4_2}\\
0 \neq & \omega^{i_{a_3}} \oplus \omega^{i_{b_3}} \oplus \omega^{i_{b_1}+i_{a_2}} \oplus \omega^{i_{b_2}+i_{a_1}}.\label{eq:prod_4_3}
\end{align}}
\end{enumerate}
\hfill$\triangle$
\end{lemma}

The proof of Lemma~\ref{lem:product_prev} uses a similar technique as used in the proof of Lemma~\ref{lem:super}, and it has therefore been omitted.

\section{Greedy algorithm}
We present a greedy algorithm for an~$n \times n$ superregular lower triangular Toeplitz matrix.
The algorithm is illustrated in Algorithm~\ref{alg:greedy_search_with_back_tracking}.
The algorithm starts by searching for a~$2 \times 2$ superregular matrix.
When a~$l \times l$ superregular matrix is found, the algorithm will search for a~$l + 1 \times l + 1$ superregular matrix by extending the~$l \times l$ matrix using the~$\phi$--operator and~$i_l$.
\begin{algorithm}
\caption{Greedy search with backtracking for an~$n \times n$ superregular lower triangular Toeplitz matrix}
\label{alg:greedy_search_with_back_tracking}
\begin{algorithmic}[1]
\REQUIRE $n \geq 2$,~$\omega \in \Omega_{f_p}$,~$A_1 = 1$
\STATE $l = 2$
\WHILE {$l \leq n$}
\STATE $h := 0$
\WHILE {$h < 2^p - 1$}
\STATE $i_{l-1} := h$
\STATE $A_l := \phi_{\omega,h}\left(A_{l-1}\right)$
\STATE Define~$\mathcal{A}_l$ using \eqref{eq:set_of_all_submatrices}
\IF {$\nexists A' \in \mathcal{A}_l$ such that~$\det(A') = 0$}
\STATE $l := l + 1$
\STATE \textbf{go to 2}
\ENDIF
\STATE $h := h + 1$
\ENDWHILE
\IF {$l = 2$}
\RETURN \textit{Insufficient field size}
\ELSE
\STATE $l := l - 1$
\STATE $h := i_{l-1} + 1~$
\STATE \textbf{go to 4}
\ENDIF
\ENDWHILE
\RETURN $A_n$
\end{algorithmic}
\end{algorithm}
The search is implemented by having~$i_l$ running through all the elements of the finite field, except the last element.
The last element,~$2^p-1$, is excluded since~$\omega^0 = \omega^{2^p-1}$, where~\mbox{$\omega \in \Omega_{f_p}$}.
This method is used until an~$n \times n$ superregular matrix is found, provided that the field size is sufficiently large.
If~\mbox{$\nexists A_{l+1}$} such that \mbox{$A_{l+1} = \phi_{\omega,i_l}(A_l) \in \mathcal{A}^\omega_{l+1}, i_l \in \mathcal{I}_1, l < n$} then backtracking is required.
That is, without backtracking the algorithm could reach a~$l \times l$ matrix, where~$l < n$, that cannot be extended further.
In case of such an event, then~$i_{l-1}$ is set to the next element and the resulting matrix is tested for superregularity.
Under sufficiently large field size the algorithm is guaranteed to find an~$n \times n$ superregular lower triangular Toeplitz matrix.
In the worst case, the algorithm will fail after having checked all possible combinations of~$\mathcal{I}_{j},j\in\{1,\dotsc,n-1\}$ before returning \textit{Insufficient field size}.
%
On an Intel~2.3~GHz \mbox{Core~i5~(I5--2415M)} our single threaded implementation of the algorithm requires less than~$230$~ms to find a~$9 \times 9$ superregular lower triangular Toeplitz matrix over~$\mathrm{GF}(2^8)$.
Furthermore, without backtracking our experiments show that the algorithm will at most work for~$n = 9$ over~$\mathrm{GF}(2^8)$.
\begin{align}
\mathcal{A}_l := &\Big\{{A_l}^{j_1,\dotsc,j_s}_{k_1,\dotsc,k_s} \in \left[\mathrm{GF}(2^p)\right]^{s \times s} : s = 2, \dotsc, l - 1, \\
&~~~ 1 \leq j_1 < \dotsc < j_s = l, 1 = k_1 < \dotsc < k_s \leq l, \\
&~~~ j_t \geq k_t, \forall t \Big\} \label{eq:set_of_all_submatrices}
\end{align}

\section{Examples of coding matrices}
\label{sec:examples}
We now present two superregular~$10\times 10$ matrices, where~$p=8$,~$f_p(\omega) = \omega^8 + \omega^4 + \omega^3 + \omega^2 + 1$, where~\mbox{$\omega \in \Omega_{f_p} = \{2, 4, 16, 29, 76, 95, 133, 157\}$}.
The matrices are shown in Equations \eqref{eq:10x10_M} and \eqref{eq:10x10_N}.
The two matrices have identical performance with respect to decoding capabilities, since they are both superregular.
However,~$A_{10}'$ outperforms~$A_{10}$ with respect to encoding and decoding throughput.
Our experiments of encoding and decoding data packets of~$1600$~bytes using~$A_{10}$ and~$A_{10}'$ show a throughput gain of~$22$~\%.
The gain in throughput comes from the fact that when an element equals~$1$, there is no need for multiplication during the encoding and decoding process.
Inspecting~\eqref{eq:10x10_M} and~\eqref{eq:10x10_N} reveals that~$A_{10}'$ has~\mbox{$i_2 = i_3 = 0$}, which in turn ensures that~$15$ of the matrix elements below the diagonal are~$1$.
Whereas,~$A_{10}$ has no elements below the diagonal that are~$1$.
Equations~\eqref{eq:10x10_M_first_column} and~\eqref{eq:10x10_N_first_column} show the first column of~$A_{10}$ and~$A_{10}'$ respectively, with~$\omega = 2$.
Given their structure these matrices are superregular for~\mbox{$n \leq 10$}.
\begin{align}
A_{10}  &= \psi_\omega\left(125, 35, 109, 219, 83, 177, 191, 39, 23\right) \label{eq:10x10_M}\\
A_{10}' &= \psi_\omega\left(  1,  0,   0,   3,  5,  10,  36, 86, 83\right) \label{eq:10x10_N}
\end{align}
\begin{align}
{A_{10}}^{1,\dotsc,10}_1  &= \left[1, 51, 156, 189,  86, 187, 219,  65,  53, 201\right]^T \label{eq:10x10_M_first_column}\\
{A_{10}'}^{1,\dotsc,10}_1 &= \left[1,  2,   1,   1,   8,  32, 116,  37, 177, 187\right]^T \label{eq:10x10_N_first_column}
\end{align}

In addition to the two superregular matrices, we also present two~$6 \times 6$ jointly superregular matrices.
These matrices are jointly superregular over~$\mathrm{GF}(2^8)$ using the previous~$f_p(\omega)$ and its roots~$\Omega_{f_p}$.
Furthermore, the two matrices are not only jointly superregular but they are also product preserving.
The matrices are shown in Equations~\eqref{eq:6x6_jointly_1} and~\eqref{eq:6x6_jointly_2}.
Note that the matrices have several parameters that are~0.
A consequence of the lower triangular Toeplitz structure of the matrices is that they are product preserving jointly superregular for any block of size~$n \leq 6$.
\begin{align}
A_{6}  &= \psi_\omega\left(0, 2, 5, 0, 15\right) \label{eq:6x6_jointly_1}\\
A_{6}' &= \psi_\omega\left(1, 0, 4, 9, 30\right) \label{eq:6x6_jointly_2}
\end{align}

Finally, we present two~$7 \times 7$ jointly superregular matrices, shown in Equations~\eqref{eq:7x7_jointly_1} and~\eqref{eq:7x7_jointly_2}.
These two matrices are not product preserving.
However, they are jointly superregular matrices for any~$n \leq 7$, due to the matrix structure.
\begin{align}
A_{7}  &= \psi_\omega\left(6, 0, 0,  4, 136, 133\right) \label{eq:7x7_jointly_1}\\
A_{7}' &= \psi_\omega\left(7, 2, 3, 11,  77, 157\right) \label{eq:7x7_jointly_2}
\end{align}

\section{Conclusions}
This paper has delivered explicit matrix constructions for superregular matrices.
We also presented a greedy algorithm for larger superregular matrices.
The matrix attributes \textit{joint superregularity} and \textit{product preserving joint superregularity} are defined for lower triangular matrices.
Furthermore, explicit matrix constructions for matrices with the two attributes are provided.
We demonstrated the applicability of (product preserving) jointly superregular matrices, with use--cases such as intermediate recoding or codes with a rate lower than~$\nicefrac{1}{2}$, respectively.
Both use--cases benefit greatly from optimal decoding capabilities.
We also exposed some general attributes of (jointly) superregular matrices.
All of the methods presented in this paper can be implemented on embedded devices.
The field size and matrix dimensions used in the example section are feasible even on low--power devices with limited instruction sets.
All the presented matrices still provide optimal decoding capabilities.
Finally, we showed that the parameters of a lower triangular Toeplitz superregular matrix have a significant impact on the throughput performance of an implementation.

\appendix[Proof of Lemma~\ref{lem:super}]

\begin{enumerate}
\item[i)]{The determinants of the proper submatrices of~$A_3^\omega$ are:~$\omega^{i_1}$ and~$\omega^{2i_1} \oplus \omega^{i_2}$, where~$\omega^{i_1} \neq 0$. Since~$\omega$ is primitive,~$\omega^{i}\neq \omega^j$, if~$(i,j)\in\mathcal{I}_2, i\neq j$. Thus,~\mbox{$\omega^{2i_1} \oplus \omega^{i_2} \neq 0 \Leftrightarrow 2i_1 \neq i_2$~(modulo~$2^p-1$)}.}
\item[ii)]{We only need to check the determinants of the proper submatrices that include the new element~$\omega^{i_3}$. Since~$\omega$ is primitive, it is easy to obtain~\eqref{eq:A4_eq}.}
\item[iii)]{It is easy to obtain the determinant expressions of the proper submatrices that include the new element~$\omega^{i_4}$.
These expressions contains terms on form~$\omega^i \oplus \omega^i$.
Since arithmetic operations are wrt.\  $\mathrm{GF}(2^p)$,~\mbox{$\omega^i \oplus \omega^i = 0,\forall\omega,\forall i$}, and we obtain~\eqref{eq:A5_eq} and~\eqref{eq:A5_eq_2}.}
\end{enumerate}



\begin{thebibliography}{10}
\providecommand{\url}[1]{#1}
\csname url@samestyle\endcsname
\providecommand{\newblock}{\relax}
\providecommand{\bibinfo}[2]{#2}
\providecommand{\BIBentrySTDinterwordspacing}{\spaceskip=0pt\relax}
\providecommand{\BIBentryALTinterwordstretchfactor}{4}
\providecommand{\BIBentryALTinterwordspacing}{\spaceskip=\fontdimen2\font plus
\BIBentryALTinterwordstretchfactor\fontdimen3\font minus
  \fontdimen4\font\relax}
\providecommand{\BIBforeignlanguage}[2]{{%
\expandafter\ifx\csname l@#1\endcsname\relax
\typeout{** WARNING: IEEEtran.bst: No hyphenation pattern has been}%
\typeout{** loaded for the language `#1'. Using the pattern for}%
\typeout{** the default language instead.}%
\else
\language=\csname l@#1\endcsname
\fi
#2}}
\providecommand{\BIBdecl}{\relax}
\BIBdecl

\bibitem{10.1109/ITCC.2001.918813}
D.~G. Sachs, I.~Kozintsev, M.~Yeung, and D.~L. Jones, ``{Hybrid ARQ for Robust
  Video Streaming Over Wireless LANs},'' \emph{Int. Conference on Information
  Technology: Coding and Computing}, pp. 317--321, 2001.

\bibitem{6620377}
A.~Badr, A.~Khisti, W.~tian Tan, and J.~Apostolopoulos, ``Robust streaming
  erasure codes based on deterministic channel approximations,'' in \emph{IEEE
  Int. Symposium on Inf. Theory}, July 2013, pp. 1002--1006.

\bibitem{4427233}
A.~Nafaa, T.~Taleb, and L.~Murphy, ``Forward error correction strategies for
  media streaming over wireless networks,'' \emph{IEEE Communications
  Magazine}, vol.~46, no.~1, pp. 72--79, January 2008.

\bibitem{6284055}
D.~Leong and T.~Ho, ``Erasure coding for real-time streaming,'' in \emph{IEEE
  Int. Symposium on Inf. Theory}, July 2012, pp. 289--293.

\bibitem{6692495}
J.~Krigslund, J.~Hansen, M.~Hundeb{\o}ll, D.~Lucani, and F.~Fitzek, ``Core:
  Cope with more in wireless meshed networks,'' in \emph{IEEE 77th Vehicular
  Technology Conference}, June 2013, pp. 1--6.

\bibitem{5934877}
H.~Seferoglu, A.~Markopoulou, and K.~Ramakrishnan, ``I2nc: Intra- and
  inter-session network coding for unicast flows in wireless networks,'' in
  \emph{IEEE INFOCOM}, April 2011, pp. 1035--1043.

\bibitem{6825089}
P.~Pahlevani, D.~Lucani, M.~Pedersen, and F.~Fitzek, ``Playncool: Opportunistic
  network coding for local optimization of routing in wireless mesh networks,''
  in \emph{IEEE Globecom Workshops}, Dec 2013, pp. 812--817.

\bibitem{7263352}
J.~Hansen, D.~Lucani, J.~Krigslund, M.~M\'{e}dard, and F.~Fitzek, ``Network
  coded software defined networking: enabling 5g transmission and storage
  networks,'' \emph{IEEE Commun. Mag.}, pp. 100--107, September 2015.

\bibitem{5963013}
{Heide, J. and Pedersen, M.V. and Fitzek, F.H.P. and M\'{e}dard, M.}, ``{On
  Code Parameters and Coding Vector Representation for Practical RLNC},'' in
  \emph{IEEE Int. Conference on Communications}, June 2011.

\bibitem{1023698}
R.~Smarandache, H.~Gluesing-Luerssen, and J.~Rosenthal, ``{Strongly MDS
  convolutional codes, a new class of codes with maximal decoding
  capability},'' in \emph{IEEE Int. Symposium on Inf. Theory}, 2002, p. 426.

\bibitem{1580796}
H.~Gluesing-Luerssen, J.~Rosenthal, and R.~Smarandache, ``{Strongly-MDS
  convolutional codes},'' \emph{IEEE Trans. Inf. Theory}, vol.~52, no.~2, pp.
  584--598, Feb 2006.

\bibitem{aissen1952}
M.~Aissen, I.~Schoenberg, and A.~Whitney, ``\BIBforeignlanguage{English}{{On
  the generating functions of totally positive sequences I}},''
  \emph{\BIBforeignlanguage{English}{{Journal d'Analyse Math\'{e}matique}}},
  vol.~2, no.~1, pp. 93--103, 1952.

\bibitem{Almeida20132145}
R.~Hutchinson, R.~Smarandache, and J.~Trumpf, ``{A new class of superregular
  matrices and MDP convolutional codes},'' \emph{Linear Alg. and its
  Applications}, vol. 439, no.~7, pp. 2145 -- 2157, 2013.

\bibitem{45291}
R.~Roth and A.~Lempel, ``{On MDS codes via Cauchy matrices},'' \emph{IEEE
  Trans. Inf. Theory}, vol.~35, no.~6, pp. 1314--1319, Nov 1989.

\bibitem{7282860}
R.~Mahmood, A.~Badr, and A.~Khisti, ``Convolutional codes with maximum column
  sum rank for network streaming,'' in \emph{IEEE Int. Symposium on Inf.
  Theory}, June 2015, pp. 2271--2275.

\end{thebibliography}
\end{document}